\documentclass[journal,10pt,twocolumns]{IEEEtran}

\IEEEoverridecommandlockouts

\ifCLASSINFOpdf
 \usepackage[pdftex]{graphicx}

  \DeclareGraphicsExtensions{.pdf,.jpeg,.png,.tiff,.tif}
\else
\fi

\usepackage{amsthm,amsfonts,amssymb}
\usepackage{amsmath}
\usepackage{bm}

\usepackage{lineno}

\newcommand{\Toro}[2]{\Phi_{\pmb{#1}}(#2)}

\newcommand{\CurvaT}[1]{\pmb{s}_{T_{#1}}}

\newcommand{\R}{\mathbb{R}}

\newcommand{\Vo}{\mathcal{V}}
\newcommand{\Sc}{\mathcal{S}\mathcal{C}}
\newcommand{\PP}{\mathcal{P}}
\newcommand{\cc}{\pmb{c}}
\newcommand{\bb}{\pmb{b}}

\newcommand{\uu}{\pmb{u}}

\newtheorem{teo}{Theorem}
\newtheorem{lema}{Lemma}

\newtheorem{ex}{Example}
\newtheorem{prop}{Proposition}

\begin{document}


%
\title{Curves on Flat Tori and Analog \\ Source-Channel Codes}

\author{Antonio~Campello,~\IEEEmembership{Student~Member,~IEEE,}
        Cristiano~Torezzan,
        and~Sueli~I.~R.~Costa,~\IEEEmembership{Member,~IEEE}
\thanks{A. Campello and S. Costa are with Institute of Mathematics, Statistics and Computer Science, University of Campinas, S\~ao Paulo, Brazil, 13083-759}
\thanks{C. Torezzan is with School of Applied Sciences, University of Campinas, S\~ao Paulo, Brazil, 13484-350}
\thanks{Part of this work was presented in the International Symposium of Information Theory (ISIT) 2012, Boston - MA.}}

\maketitle



%


\maketitle

\begin{abstract}
In this paper we consider the problem of transmitting a continuous alphabet discrete-time source over an AWGN channel in the bandwidth expansion case. We propose a constructive scheme based on a set of curves on the surface of a $2N$-dimensional sphere. Our approach shows that the design of good codes for this communication problem is related to geometrical properties of spherical codes and projections of $N$-dimensional rectangular lattices. Theoretical comparisons with some previous works in terms of the mean squared error as a function of the channel SNR, as well as simulations, are provided.
\end{abstract}

\section{Introduction}

The problem of designing good codes for a continuous alphabet source to be transmitted over a Gaussian channel with power constraint and bandwidth expansion can be viewed as the one of constructing curves in the Euclidean space of maximal length and such that its folds (or laps) are a good distance apart. When the channel noise is below a certain threshold, a bigger length essentially means a higher resolution when estimating the sent value. On the other hand, if the folds of the curve come too close, this threshold will be small and the mean squared error (mse) will be dominated by large errors as  consequence of decoding to the wrong fold. Explicit constructions of curves for analog source-channel coding for bandwidth expansion were presented, for example, in \cite{Sueli} and \cite{polynomial}. For insights on the bandwidth reduction case we refer the reader to \cite{bandwidthReduction}.

In this work, we consider spherical curves in the $2N$-dimensional Euclidean space. We develop an extension of the construction presented in \cite{Sueli} to a set of curves on layers of flat tori. Our approach explores the geometrical properties of spherical codes and projections of $N$-dimensional lattices. In the scheme presented here, homogeneity and low decoding complexity properties were preserved from \cite{Sueli} whereas the total length can be meaningfully increased. Theoretical results as well as simulations show that our approach outperforms the previous ones in terms of the tradeoff between SNR and mean squared error.

This paper is organized as follows: in Sections II and III, we introduce some mathematical/information-theoretical tools used in our approach. In Section IV we present a scheme to design piecewise homogeneous curves on the Euclidean sphere and describe the encoding/decoding process. In Section V we derive an extended version of the \textit{Lifting Construction} \cite{FatStrut} suitable to our problem and present some examples and length comparisons with some previous constructions. In Section VI we provide an analysis of the asymptotic behavior of the mse as a function of the channel SNR and in Section VII simulation results are exhibited.

\section{Communication Framework}
The communication system we consider here is illustrated below. The objective is to transmit an input real value $x$ associated to the random variable $X$ over a channel of dimension $N$ with additive Gaussian noise and average power constraint. The encoder uses a bandwidth expansion mapping in order to encode $x$ to a point $\pmb{s}(x) \in \mathbb{R}^N$ respecting the constraint $E \left[\pmb{s}(X) \pmb{s}(X)^t\right] \leq P$. The decoder recovers an estimate $\hat{x}$ for the sent value attempting to minimize the mean squared error (mse) $E[ (X-\hat{X})^2]$. The source is assumed to have pdf with support $\left[0, 1\right)$ and $\pmb{s}([0,1))$ will be referred to as the \textit{signal locus}. If $\pmb{s}$ is a continuous mapping, then the signal locus is a curve in $\mathbb{R}^N$. 

\begin{figure}[!htb]
\centering
\includegraphics[scale=0.65]{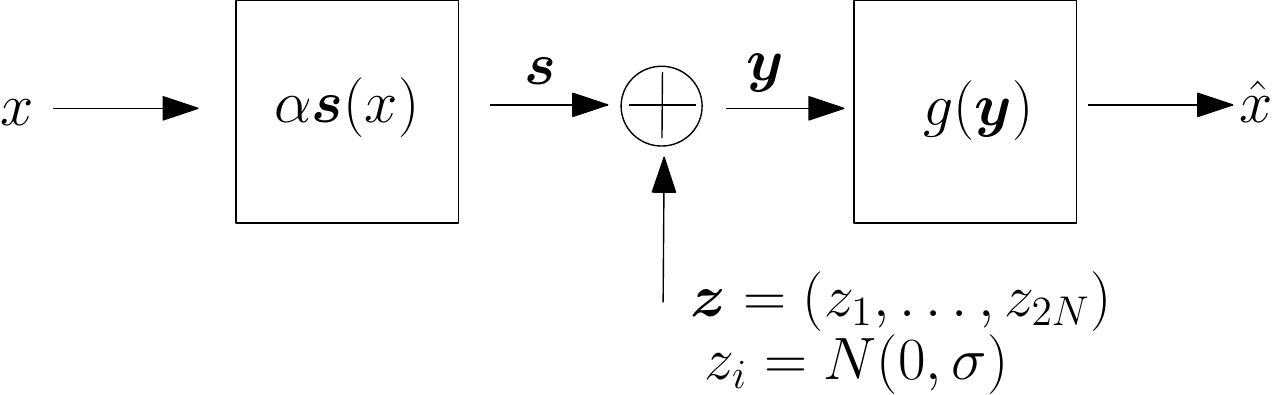}
\caption{Communication system}	
\label{fig:com}
\end{figure}

This communication model is carefully analyzed in \cite[Ch. 4]{Sakrison}. Given a radius $\rho > 0$, let $E_\rho$ denote the event $
\left\{\left\| \pmb{z} \right\| < \rho\right\}$. The mse can be split into two terms such that:
\begin{equation} \mbox{mse} = E[(X-\hat{X})^2|E_\rho]P(E_\rho) + E[(X-\hat{X})^2|E_\rho^c]P(E_\rho^c).
\end{equation}

For high values of the SNR $P/\sigma^2$, $P(E_\rho^c) \approx 0$ and the mse is well approximated by $E[(X-\hat{X})^2|E_\rho]$ (the goodness of the approximation depends on the SNR and $\rho$). In this case, called the \textit{low noise regime}, it is possible to show \cite{Sakrison} that:

\begin{equation}
E[(X-\hat{X})^2] \approx \sigma^2 \int_0^1 p(x) \left\| \dot{\pmb{s}}(x) \right\|^{-2}dx := E_{\small \mbox{low}}[(X-\hat{X})^2],
\label{eq:lownoise}
\end{equation}
where $p(x)$ is the pdf of the source. Most of the analysis in Section \ref{sec:mse} will be done under the low noise regime. If the source is uniform, then an optimal parametrization for $\pmb{s}(x)$ yields 

\begin{equation}
E_{\small \mbox{low}}[(X-\hat{X})^2] = \frac{\sigma^2}{L^2},
\end{equation}
where $L$ is the curve length.

A geometric picture of the possible effects of the noise on the estimated value is given by Figure 1. Under the low noise regime, all the calculations can be essentially done by approximating $\pmb{s}(x)$ by its tangent line and the system will be described by Equation \eqref{eq:lownoise}. However, this is not true for large errors. Decoding to the wrong fold of the curve yields a poor performance of the system. This suggests that our objective is to choose long curves with a large distance between its folds.

In our constructions, the signal locus will be a collection of curves in $\mathbb{R}^N$ rather than a single curve (this can be viewed as a ``piecewise continuous'' curve). Our analysis is mainly focused on uniform sources, although the scheme can be easily adapted to any source with limited support and even to Gaussian sources, by means of a companding function \cite{Sakrison}, as shown in Section VII.

\begin{figure}[!htb]
\centering
\includegraphics[scale=0.55]{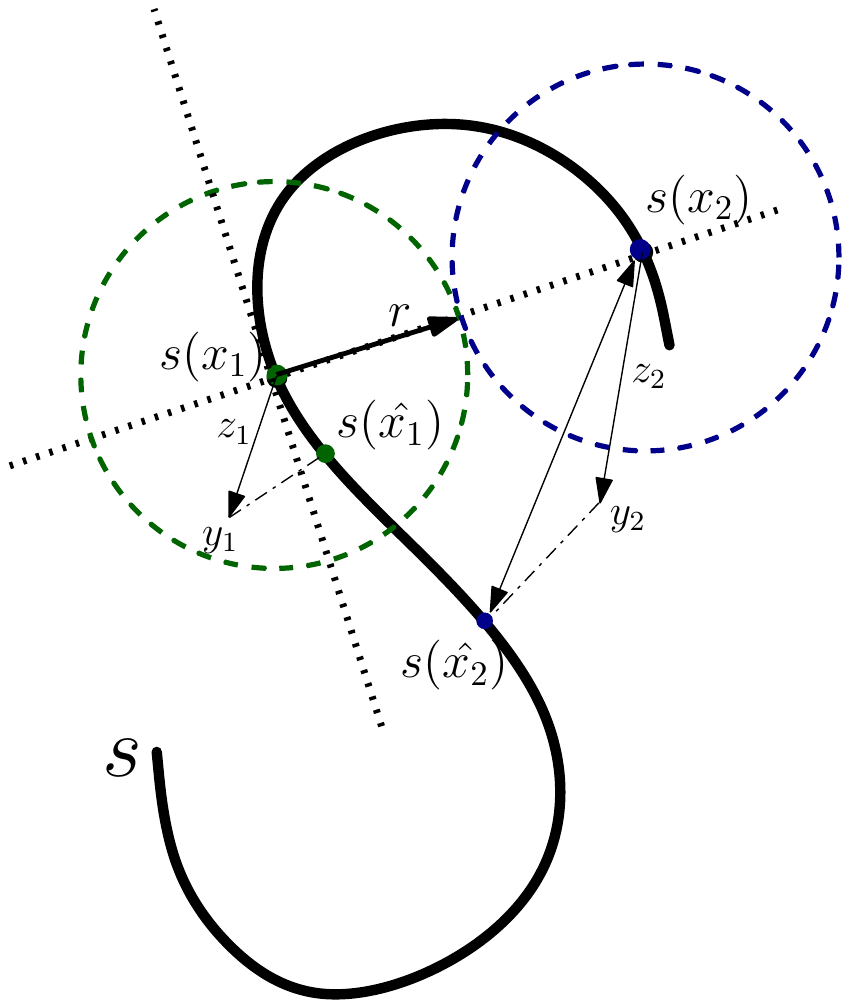}
\caption{Small errors and large errors.}
\label{fig:noise}
\end{figure}

\section{Background}
\label{sec:background}
\subsection{Flat Tori}
The unit sphere $S^{2N-1} \subset \R^{2N} $ can be foliated with flat tori \cite{BergerGostiaux, Cristiano} as follows. For each unit vector ${\cc = (c_{1},c_{2},..,c_{N}) \in \R^{N}}, c_i > 0$, and ${\pmb{u}=(u_1,u_2,\ldots,u_N) \in \R^N}$, let $\Phi_{\cc}:\R^N \rightarrow \R^{2N}$ be defined by
\begin{equation}
\small{
\Phi_{\cc} (\pmb u)=\left(c_{1}\left(\cos \frac{u_{1}%
}{c_{1}},\sin \frac{u_{1}}{c_{1}}\right),\dots,c_{N}\left(\cos \frac{u_{N}}{c_{N}},\sin \frac{u_{N}}{c_{N}}\right)\right).
}
\label{eq:Toro}
\end{equation}
The $N$-periodic function $\Phi _{\cc}$ is a local isometry on its image, the torus $ T_{\cc}$, a flat $N$-dimensional surface contained in the unit sphere $S^{2N-1}$. The torus $ T_{\cc} = \Phi_{\cc}(R^N)$ is also the image of the hyperbox:
\begin{equation}
\label{para}
\PP_{\cc} := \{\uu \in \R^{N}; 0 \leq u_{i} < 2 \pi c_{i}\}, \ \ 1\leq i\leq N.
\end{equation}

Note also that each vector of $S^{2N-1}$ belongs to one, and only one, of these flat tori if we also consider the degenerated cases where some $\cc_i$ may vanish. Thus we say that the family of flat tori $T_{\cc}$ and their degenerations, with $\cc = (c_{1},c_{2},..,c_{N})$, $ \left\|  \cc  \right\|  =1$, $c_{i}  \geq 0$, defined above is a foliation of the unit sphere of $S^{2N-1}\subset \R^{2N}.$

It can be shown (Proposition 1 in \cite{Cristiano}) that the minimum distance between two points, one in each flat torus $T_{\bb}$ and $T_{\cc}$, is 
\begin{equation}
\label{eq:distdoistoros}
d(T_{\cc},T_{\bb})= \left\|  \cc- \bb \right\| = \left( \sum_{i=1}^N (c_i - b_i)^2\right)^{1/2}.
\end{equation}

The distance between two points on the same torus $T_{\pmb{c}}$, given by
\begin{equation*}
d(\Phi_{\pmb{c}}(\pmb{u}),\Phi_{\pmb{c}}(\pmb{v}))=2\sqrt{\sum c_{i}^{2}\sin^{2}\left(\frac{u_{i}-v_{i}%
}{2c_{i}}\right)}\label{SameTorus}
\end{equation*}
is bounded in terms of $\|\pmb{u}-\pmb{v}\|$ by the following proposition.
\begin{prop}\cite{Cristiano}
Let $\pmb{c=}(c_{1},c_{2},..,c_{N})$, $ \left\|  \cc  \right\|  =1$,
and let  $\pmb{u}, \pmb{v} \in \PP_{\cc}$. Let $\Delta =  \left\|\pmb{u}-\pmb{v}\right\|$ and $\delta= \left\|\Phi_{\pmb{c}}(\pmb{u})- \Phi_{\pmb{c}}(\pmb{v}) \right\|$.
Then
\begin{equation}
\displaystyle 
\frac{2}{\pi}\Delta\leq\dfrac{\sin\frac{\Delta}{2c_{\xi}}}{\frac{\Delta}{2c_{\xi}%
}}\Delta\leq\delta\leq\dfrac{\sin\frac{\Delta}{2}}{\frac{\Delta}{2}}\Delta
\leq\Delta
\label{eq:FamousRelation}
\end{equation}
where $\displaystyle c_{\xi} = \min c_i$.
\\
\end{prop}

\subsection{Curves}
\label{subsec:curves}
As pointed out earlier in \cite{Sueli}, some important properties of a curve from a communication point of view are its \textit{stretch} and \textit{small-ball radius}. Given a curve $\pmb{s}: [a,b] \mapsto \mathbb{R}^N$, the Voronoi region $V(x)$ of a point $\pmb{s}(x)$ is the set of all points in $\mathbb{R}^N$ which are closer to $\pmb{s}(x)$ than to any other point of the curve. If $H(x)$ denotes the hyperplane orthogonal to the curve at $\pmb{s}(x)$, then the maximal \textit{small-ball radius} of $\pmb{s}$ is the largest $r > 0$ such that $B_r(\pmb{s}(x)) \cap H(x) \subset V(x)$ for all $x \in [a,b]$, where $B_r(\pmb{s}(x))$ is the Euclidean ball of radius $r$ centered at $\pmb{s}(x)$. This means that the ``tube'' of radius $r$ placed along the curve does not intersect itself. The \textit{stretch} $\mathcal{S}(x)$ is the function $\left\| \dot{\pmb{s}}(x) \right\|$ where $\dot{\pmb{s}}(x)$ is the derivative of $\pmb{s}(x)$. The length of a curve is given by $\int_a^b \mathcal{S}(x) dx$. In this paper we will be interested in curves with large length and small-ball radius.

\section{The Torus Layer Scheme}
The design of those curves in this section is essentially divided in two parts. First, we choose a collection of $M$ tori on the surface of the Euclidean sphere $S^{2N-1}$ at least $2 \delta$ apart. The approach for doing this is via discrete spherical codes in $\mathbb{R}^N$. Second, we show a systematic way of constructing curves on each layer, via projection lattices in $ \mathbb{R}^{N-1}$. Finally, we give a description of the whole signal locus and summarize the encoding process.

For the Torus Layer Scheme, the unit interval will be partitioned into $M$ intervals of different length, and each of them mapped into a curve on one of the layers. It is worth noticing that, for the special case $M=1$, if we choose the torus associated to the vector $c = \frac{1}{\sqrt{N}} (1,\ldots,1)$ to encode the information, then the scheme proposed here is exactly the one analysed in \cite{Sueli}. As we show next, for $M > 1$ the curves presented here outperform the ones presented in \cite{Sueli} (and also in \cite{FatStrut}) in terms of length and small-ball radius. 
\subsection{Torus Layers}
\label{sec:torusLayers}
Given a target small ball radius $\delta$, the first step of our approach is to define a collection $T=\left\{T_1, T_2, \cdots, T_M \right\}$ of flat tori on $S^{2N-1}$ such that the minimum distance \eqref{eq:distdoistoros} between any two of them is greater than $2\delta$. This step is equivalent to design a $N$-dimensional spherical code $\Sc_+$ with minimum distance $2\delta$ and such that its $M$ codewords have non-negative coordinates.

Each point $\cc \in \Sc_+$ defines a hyperbox $\PP_{\cc}$ (\ref{para}) and hence a flat torus $T_{\cc}$ in the unit sphere $S^{2N-1}$. By \eqref{eq:distdoistoros} we can assert that the distance between two of those tori is at least $2 \delta$.

There are several ways of constructing spherical codes that can be employed here, e.g. \cite{eric,Hamkins1} or even on layers of flat tori as introduced in \cite{Cristiano}. We present next an example of a family of discrete spherical codes in $\mathbb{R}^N$ that will be useful in Section \ref{sec:mse}.
\begin{ex}
Let
\begin{equation}
\displaystyle
\pmb{c}(t) = \frac{ (1,1+t,1+2t,\cdots, 1+(N-1)t)}{\sqrt{\sum_{i=0}^{N-1}(1+it)^2}}, \mbox{ } t>0
\label{eq:torosPerm}
\end{equation}
be a non-negative $N$-dimensional unit vector.

The set $\mathcal{S} \mathcal{C}_{c(t)} =  \{ \sigma(\pmb{c}(t)): \sigma \in S_N \}$ of all permutations of $\pmb{c}(t)$ defines a $N$-dimensional spherical code with non-negative coordinates, minimum distance equals to $$d(t)=\frac{t\sqrt{2}}{\sqrt{\sum_{i=0}^{N-1}(1+it)^2}},$$ and cardinality $M = N!$. This set can be used in the first step of the construction described in this section.

Note that  
\begin{equation} \displaystyle \lim_{t\to 0} \pmb{c}(t) = \frac{1}{\sqrt{N}} (1,\ldots,1) := \pmb{\hat{e}},
\end{equation}
and, for each $t > 0$, the codewords of $\mathcal{S} \mathcal{C}_{c(t)}$ are equidistant from $t\pmb{\hat{e}}$, the vector that determines the maximum volume torus $T_{\pmb{\hat{e}}}$.

Moreover, for any fixed dimension $N$ and $d_0 \in \mathbb{R}$ sufficiently small, there exists $t >0$ such that $d(t) = d_0$. This follows by observing that the equation

$$d_0=\frac{t\sqrt{2}}{\sqrt{\sum_{i=0}^{N-1}(1+it)^2}}$$
has a positive root for

$$0 < d_0 < \frac{2 \sqrt{3}}{\sqrt{(N-1) N (2 N-1)}}.$$
\end{ex}

Although the construction above is not the best one in terms of number of points for fixed $d$, it has the advantage of being highly symmetric and having a closed form for the spherical code minimum distance.

\subsection{Curves on Each Torus}
Let $T_{\pmb{c}}$ be a torus represented by the vector $\pmb{c} \in \Sc_+$ as defined in the previous section. On the surface of $T_{\pmb{c}}$ we consider curves of the form: 
\begin{equation}
\CurvaT{c}(x) = \Toro{c}{x 2\pi \hat{\pmb{u}}},
\label{def:Curva}
\end{equation}
where $C = \mbox{diag}(c_1,\ldots,c_N)$, $\hat{\pmb{u}} = \pmb{u}C = (c_1 u_1, \ldots, c_N u_N)$, $\Phi_{\pmb{c}}$ is given by \eqref{eq:Toro} and $x \in [0,1]$.

Provided that $\pmb{u} \in \mathbb{Z}^N, \gcd(u_i) = 1$, those curves are closed (a $(u_1,\ldots,u_N)$-type knot), and due to periodicity and local isometry properties of $\Phi_{\pmb{c}}$ their lengths are $2 \pi \left\| \hat{\pmb{u}} \right\|$. They are also the image through $\Phi_{\pmb{c}}$ of the intersection between the set of lines $2 \pi W$, $W = \left\{ \hat{\pmb{u}}x + \hat{\pmb{n}} : \hat{\pmb{n}} = \pmb{n} C \mbox{ and } \pmb{n} \in \mathbb{Z}^{N} \right\}$, and the box $\PP_{\cc}$. For  $c = \hat{\pmb{e}}$, these are exactly the curves analyzed in \cite{Sueli}.

Let $r_{\pmb{c}}(\pmb{u})$ be the minimum distance between two different lines in $W$. We have:
\begin{equation}
\begin{split}
r_{\pmb{c}}(\pmb{u}) &:= \min_{\hat{\pmb{n}} \neq k \pmb{\hat{u}}, k \in \mathbb{Z}} \min_{\hat{x}, x} \left\| \hat{\pmb{u}}x - (\hat{\pmb{u}}\hat{x}+\hat{\pmb{n}}) \right\| \\
& = \min_{\hat{\pmb{n}} \neq k \pmb{\hat{u}}, k \in \mathbb{Z}} \min_{x} \left\| \hat{\pmb{u}}x - \hat{\pmb{n}} \right\| \\
& = \min_{\hat{\pmb{n}} \neq k \pmb{\hat{u}}, k \in \mathbb{Z}} \left\| P_{\hat{\pmb{u}}^{\perp}} (\hat{\pmb{n}}) \right\| \\
& = \min_{\hat{\pmb{n}} \notin \hat{\pmb{u}}^{\perp}} \left\| P_{\hat{\pmb{u}}^{\perp}} (\hat{\pmb{n}}) \right\|,
\end{split}
\end{equation}
where $P_{\hat{\pmb{u}}^{\perp}} (\hat{\pmb{n}})$ denotes the orthogonal projection of $\hat{\pmb{n}}$ onto the hyperplane $\hat{\pmb{u}}^{\perp}$ which is given by the standard projection formula

\begin{equation}
P_{\hat{\pmb{u}}^{\perp}} (\hat{\pmb{n}}) = \hat{\pmb{n}} \left(I_N - \frac{\pmb{\hat{u}}^t \pmb{\hat{u}}}{\pmb{\hat{u}} \pmb{\hat{u}}^t}\right).
\label{eq:proj}
\end{equation}

If we consider  $\Lambda_{\pmb{c}} = c_1 \mathbb{Z} \oplus \ldots \oplus c_N \mathbb{Z}$, the rectangular lattice generated by matrix $C$, then $ r_{\pmb{c}}(\pmb{u})$ is the length of shortest non-zero vector of the projection\footnote{In general, the projection of a lattice $\Lambda$ onto a subspace $H$ is \textit{not} a lattice unless certain special conditions are met, e.g., when $H^\perp$ is spanned by primitive vectors of $\Lambda$ \cite{Perfect}. This will be always the case in this paper, since $H = \hat{\pmb{u}}^\perp$ for a primitive vector $\hat{\pmb{u}}$.} of $\Lambda_{\pmb{c}}$ onto $\hat{\pmb{u}}^{\perp}$. Due to \eqref{eq:FamousRelation}, the small-ball radius $\delta_{\pmb{u},\pmb{c}}$ of $\pmb{s}_{T_{\pmb{c}}}$ can be bounded in terms of $r_{\pmb{c}} (\pmb{u})$ as follows:

\begin{equation}
2 c_{\xi} \sin \left({\frac{\pi  r_{\pmb{c}}(\pmb{u})}{2 c_{\xi}}}\right) \leq \delta_{\pmb{u},\pmb{c}} \leq 2 \sin \left({\frac{\pi r_{\pmb{c}}(\pmb{u})}{2}}\right),
\label{eq:distortionCurve}
\end{equation}
where $c_{\xi} = \displaystyle \min_{i}{c_i}$ and $c_i > 0$. Thus, for small values of $\delta_{\pmb{u}, \pmb{c}} $, we have $\delta_{\pmb{u},\pmb{c}} \approx  \pi r_{\pmb{c}}(\pmb{u})$. Our goal is to choose $\pmb{u}$ in order to maximize $r_{\pmb{c}}(\pmb{u})$. In addition, we also want to reach a contrary objective, which is the one of maximizing the arc length $l_{\pmb{u},\pmb{c}} = 2 \pi \left\|\hat{\pmb{u}} \right\|$ of $\pmb{s}_{T_{\pmb{c}}}$. 

The proportion (density) of the volume of the ``tube'' of radius $\pi r_{\pmb{c}}(\pmb{u})$ inside $\PP_{\pmb{c}}$ is precisely the packing density of the lattice  $P_{\hat{\bm{u}}^{\perp}} (\Lambda_{\bm{c}})$, the projection of $\Lambda_{\bm{c}}$ onto $\hat{\bm{u}}^{\perp}$, which is given by \cite{Perfect}:

\begin{equation}
\Delta(P_{\hat{\bm{u}}^{\perp}} (\Lambda_{\bm{c}})) = \Vo_{N-1} \left( \frac{r_{\pmb{c}}(\pmb{u})^{N-1} \left\|\hat{\pmb{u}} \right\|}{2^{N-1} \prod_{i=1}^N c_i} \right) \leq \Delta_{N-1}
\label{eq:BoundDensidade}
\end{equation}
where $\Delta_{N-1}$ is the density of the best lattice in dimension $(N-1)$ and $\Vo_{N-1}$ is the volume of the $(N-1)$-dimensional unit sphere. For the case when all entries of $\bm{c}$ are equal, $\Lambda_{\bm{c}}$ is equivalent to $\mathbb{Z}^N$ and it was shown in \cite{FatStrut} that we can make the above bound as tight as we want. We will show in Section V that this is also true for an arbitrary $\bm{c}$ i.e., that projections of the rectangular lattice $\Lambda_{\pmb{c}}$ can also yield dense lattice packings and therefore we can construct curves on the flat torus with the parameters arbitrary close to this bound.

\begin{ex} Let $N = 2$. Consider the local isometry
\begin{equation}
\Phi_{\cc} (\pmb u)=\left(c_{1} \cos \frac{u_{1}
}{c_{1}}, c_1 \sin \frac{u_{1}}{c_{1}} ,c_{2}\cos \frac{u_{2}}{c_{2}},c_2 \sin \frac{u_{2}}{c_{2}}\right)
\end{equation}
on the flat torus $T_{\pmb{c}}$ and the line segment given by   $\pmb{v}(x) = x \pmb{v} = x (2 \pi u_1 c_1, 2\pi u_2 c_2)$, $,u_1, u_2 \in \mathbb{Z} = x 2 \pi \pmb{\hat{u}}$ and $0 \leq x \leq 1$. The curve $\CurvaT{c}(x)$ will be the composition $\Phi(\pmb{v}(x))$ and we have
\begin{equation} 
r_{\pmb{c}}(\pmb{v}) \left\| \pmb{v} \right\| = 2\pi c_1 c_2 \Rightarrow  r_{\pmb{c}}(\pmb{v}) = \frac{c_1 c_2}{\sqrt{ u_1^2 c_1^2+ u_2^2 c_2^2}}
\end{equation}

This curve in $\mathbb{R}^4$ turns around $u_1$-times the circle obtained by its projection on the first two coordinates, whereas turning around $u_2$-times the circle of radius $c_2$ given by its last two coordinates (a type $(u_1,u_2)$ knot on the flat torus $T_{\pmb{c}}$). In this case, it is easy to calculate $r_{\pmb{c}}(\pmb{u})$ since the density of the projection lattice is equal to one. In Figure \ref{fig:encodingprocess} it is illustrated the curve $(u_1,u_2)=(4,5)$, with $c_1 > c_2$.
\end{ex}

\subsection{Encoding}
\label{subsec:Enc}
Let $T = \left\{ T_1, \ldots, T_M \right\}$ be a collection of $M$ tori as designed in Section \ref{sec:torusLayers}. For each one of these tori, let $\CurvaT{k}(x) = \Phi(x 2\pi \hat{\pmb{u}}_k), (k=1,2,\ldots,M)$ be the curve on $T_k$, determined by the vector $\hat{\pmb{u}}_k$ \eqref{def:Curva} and  consider $L = \sum_{j=1}^M L_j$, where $L_k$ is the length of $\CurvaT{k}$.

Now split the unit interval $[0,1]$, excluding one of the endpoints, into $M$ pieces according to the length of each curve:
\begin{equation*}
\left[0,1\right) = I_1 \cup I_2, \ldots \cup I_M\mbox{, where}
\end{equation*}
\begin{equation*}
I_k = \left[\frac{\sum_{j=1}^{k-1} L_j}{L},\frac{\sum_{j=1}^{k} L_j}{L} \right), \mbox{ for } k = 1, \ldots, M.
\end{equation*}
and consider the bijective mapping
$$\begin{array}{cc}
f_k: I_k  \to [0,1) \\
f_k(x) = \displaystyle\frac{x - \sum_{j=1}^{k-1} L_j/L}{L_k/L}.
\end{array}
$$
Then the full encoding map $\pmb{s}$ can be defined by
\begin{equation}
\pmb{s}(x) := \CurvaT{k}(f_k(x)), \mbox{ if } x \in I_k.
\label{eq:encoding}
\end{equation}
 The stretch of $\pmb{s}$ is constant and equal its total length $L$ and the small-ball radius of $\pmb{s}$ is the minimum small-ball radius $\delta$ of the curves $\CurvaT{k}$, provided that the distance between any pair of torus in $T$ is at least $2\delta$.


To encode a value $x$ within $[0,1]$ we apply the map \eqref{eq:encoding}, with a scaling factor of $\sqrt{P}$ such that the average transmission power is $P$. The signal locus will be a set of $M$ closed curves, each one lying on a torus layer $T_k$ and defined by a vector $\pmb{u}_k$. This whole process is illustrated in Figure \ref{fig:encodingprocess}.

\begin{figure}[!htb]
\centering
\includegraphics[scale=0.7]{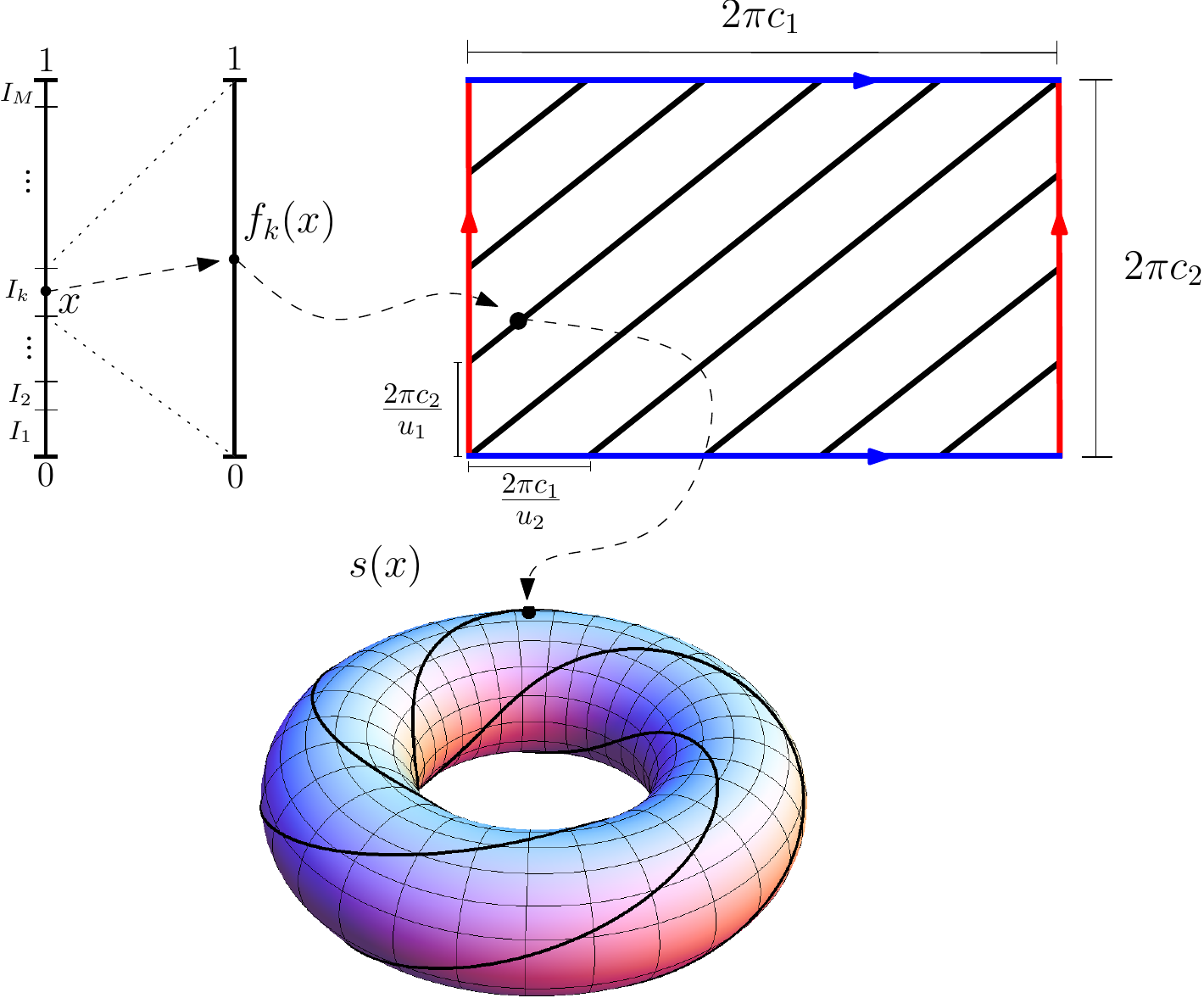}
\caption{Encoding Process}
\label{fig:encodingprocess}
\end{figure}

If the source is uniformly distributed over $[0,1]$, the encoding process presented above is a proper one, since all subintervals are equally stretched (see Section \ref{sec:mse}). For other applications, however, it could be worth considering another partition.
\subsection{Decoding}
Given a received vector $\pmb{y} \in \mathbb{R}^{2N}$, the maximum likelihood decoding is finding $\hat{x}$ such that:

$$\hat{x} = \mbox{arg}\min_{x \in [0,1]} \left\|\pmb{y} - \pmb{s}(x)\right\|.$$

Since exactly solving this problem is computationally expensive we focus on a suboptimal decoder.

For $ 0 \neq \gamma_i = \sqrt{y_{2i-1}^2+y_{2i}^2} $, we can write
{\small
\begin{eqnarray}
\pmb{y} & = &  \left( \gamma_1 \left( \frac{y_1}{\gamma_1}, \frac{y_2}{\gamma_1} \right), \hdots , \gamma_N \left( \frac{y_{2N-1}}{\gamma_N}, \frac{y_{2N}}{\gamma_N} \right)
 \right) \nonumber \\
\pmb{y} & = &  \left( \gamma_1\left(\cos{\frac{\theta_{1}}{\gamma_1}}, \sin{\frac{\theta_{1}}{\gamma_1}} \right), \hdots ,
\gamma_N \left(\cos{\frac{\theta_{N}}{\gamma_N}},\sin{ \frac{\theta_{N}}{\gamma_N}} \right) \right), \nonumber 
\end{eqnarray}
}
where,
 \begin{eqnarray}
  \theta_{i} & = & \arccos{ \left( \dfrac{y_{2i-1}}{\gamma_i} \right) } \gamma_i,  \ \ 1 \leq i \leq N \nonumber.
 \end{eqnarray}

The process of finding the closest layer involves a $N$-dimensional spherical decoding of $\pmb{\gamma}=(\gamma_1,\gamma_2, \cdots, \gamma_N)$, which has complexity $  O(MN)$.

Let $\pmb{c}_i = (c_{i1}, c_{i2}, \cdots, c_{iN})$ be the closest point in $\in \Sc_+$ to $\pmb{\gamma}$ and
$$\bar{\pmb{y}}_i = \left( c_{i1}\left(\cos{\frac{\theta_{1}}{\gamma_1}}, \sin{\frac{\theta_{1}}{\gamma_1}} \right), \hdots ,
c_{iN} \left(\cos{\frac{\theta_{N}}{\gamma_N}},\sin{ \frac{\theta_{N}}{\gamma_N}} \right) \right)$$ be the projection of $\pmb{y}$ in the torus $T_{c_i}$, i.e., 
$$
||\pmb{y}-\bar{\pmb{y}}_i|| \leq || \pmb{y} - \pmb{v} || \, ,  \forall \, \pmb{v} \in T_{c_i}.
$$
From now on, we proceed the process by using a slight modification of the torus decoding algorithm \cite{Sueli} applied to the $N$-dimensional hyperbox $\mathcal{P}_{\cc_i}$. The complexity of this algorithm is given by $O(N \left\| \pmb{u}_i \right\|_1)$, where $\pmb{u}_i$ is the vector that determines the curve $\CurvaT{i}$. Hence, if $M = O( \max_i{\left\| \pmb{u}_i \right\|_1)}$, the overall complexity of the process described in this section is $O(N \max_i{\left\| \pmb{u}_i \right\|_1)}$, the same as for the torus decoding.

\section{A Scaled Lifting Construction}
\subsection{The construction}
The \textit{Lifting Construction} was proposed in \cite{FatStrut} as a solution to the problem of finding dense lattices which are equivalent to orthogonal projections of $\mathbb{Z}^N$ along integer vectors (``fat-strut'' problem). In this section we adapt that strategy in order to construct projections of the lattice $\Lambda_{\pmb{c}}$ which approximate any $(N-1)$ dimensional lattice (hence the densest one) with the objective of finding curves in $\mathbb{R}^N$ approaching the bound \eqref{eq:BoundDensidade}. We adopt here the lattice terms as in \cite{SloaneLivro}. For our purposes, the proximity measure for lattices will be the distance between their Gram matrices, as in \cite{FatStrut}. This notion measures how close a lattice is to another one up to congruence transformations (rotations or reflections).

We consider the dual of a lattice $\Lambda \in \mathbb{R}^N$
$$\Lambda^* = \{ \pmb{x} \in \mbox{span}(\Lambda): \langle \pmb{x}, \pmb{y} \rangle \in \mathbb{Z} \,\, \forall \pmb{y} \in \Lambda\},$$ where $\mbox{span}(\Lambda)$ is the subspace spanned by a basis of $\Lambda$. Now let $\Lambda_{\pmb{c}} = c_1 \mathbb{Z} \oplus \ldots \oplus c_N \mathbb{Z}$ be the rectangular lattice generated by the diagonal matrix $C$ (and with Gram matrix $CC^t = C^2$). By scaling $\Lambda_{\pmb{c}}$, we may assume that $c_1 = 1$. With this condition, if $P_{\pmb{u}^{\perp}}(\Lambda_{\pmb{c}})^*$ denotes the dual of the projection of $\Lambda_{\pmb{c}}$ through a vector $\hat{\pmb{u}}=(1, u_2 c_2 \ldots, u_N c_N)$ ($u_i \in \mathbb{Z}$), then a generator matrix for $P_{\hat{\pmb{u}}^{\perp}}(\Lambda_{\pmb{c}})^*$ is given by

\begin{equation}
M = \begin{pmatrix}-u_2 & 1/c_2 & 0 & \ldots & 0 \\
-u_3 & 0 & 1/c_3 & \ldots & 0 \\
\vdots & \vdots & \vdots & \ddots & \vdots \\
-u_n & 0 & 0 & \ldots & 1/c_n.
\end{pmatrix}
\label{matrizBoa},
\end{equation}
what can be derived as a consequence of Prop. 1.3.4 \cite{Perfect}. In what follows we derive a general construction of projections such that $P_{\pmb{u}^{\perp}}$ is arbitrarily close to a target lattice $\Lambda \in \mathbb{R}^{N-1}$.

\begin{teo} Let $\Lambda_{\pmb{c}} = \mathbb{Z} \oplus c_2 \mathbb{Z} \ldots \oplus c_N \mathbb{Z}$, $c_i \in \mathbb{R}$. Let $\Lambda$ be a target lattice in $\mathbb{R}^{N-1}$ and consider a lower triangular generator matrix  $L^* = (l_{ij}^*)$ for $\Lambda^*$. If $\Lambda_w^*$, $w \in \mathbb{N}$ is the sequence of lattices generated by the matrices

\begin{equation}
L_w^* := \left(\begin{array}{ccccc}
\lfloor{w l_{11}^*}\rfloor & \frac{1}{c_2} & \ldots & \ldots & 0 \\
\lfloor{w l_{21}^*} \rfloor & \frac{\lfloor{w l_{22}^* c_2}\rfloor }{c_2} & \ldots & \ldots & 0 \\
\vdots & \vdots & \ddots & \ldots & 0 \\
\lfloor{w l_{n1}^*} \rfloor & \frac{\lfloor{w l_{n2}^* c_2} \rfloor }{c_2} & \ldots & \frac{\lfloor{c_{n-1} w l_{nn}^*} \rfloor}{c_{n-1}} & \frac{1}{c_n}
\end{array}\right),
\label{matrizCons}
\end{equation}
then:
\begin{enumerate}
\item[(i)] $L_w^* = P_{\hat{\pmb{u}}^{\perp}}(\Lambda)^*$ for some $\hat{\pmb{u}} \in \mathbb{R}^N$ and
\item[(ii)] $(1/w^2) L_w^* L_w^{*t} \to L^* L^{*t}$ as $w \to \infty$.
\end{enumerate}
\end{teo}
\begin{proof} Applying elementary (integer) operations on $L_w^*$ we can put it on form \eqref{matrizBoa} for some integers $u_2,\ldots,u_n$ depending on $w$, hence $L_w^*$ is a generator matrix for $P_{\hat{\pmb{u}}^{\perp}}(\Lambda)^*$, proving the first statement.
For the second statement, we clearly have $(1/w) L_w^* \to \left[ L^* \,\,\,\,\, \pmb{0} \right]$ as $w \to \infty$, where $\pmb{0}$ is the $(n-1) \times n$ all-zero column vector. Therefore, $(1/w^2) L_w^* L_w^{*t} \to \left[ L^* \,\,\,\,\, \pmb{0} \right]\left[ L^* \,\,\,\,\, \pmb{0} \right]^t = L^* L^{*t}$.
\end{proof}

Since the density of a lattice is a continuous function of the entries of its Gram matrix, it follows that the sequence of curves produced by the vectors $\pmb{\hat{u}}$ approaches the bound \eqref{eq:BoundDensidade} as $w \to \infty$.

\begin{ex} Consider the hexagonal lattice \cite{SloaneLivro}, which is the best packing in two dimensions and is equivalent to its dual. One of its generator matrix is

$$L^* = \left(\begin{array}{cc} 1 & 0 \\ \frac{1}{2} & \frac{\sqrt{3}}{2} \end{array} \right).$$

We apply the construction above and reduce, through elementary operations, the matrix $L_w^*$ in order to put it on form \eqref{matrizBoa}. After re-scaling the rectangular lattice $\Lambda_{\pmb{c}}$, we find the sequence of vectors

$$\hat{\pmb{u}}_w = (c_1, -2wc_2, (2w \lfloor w \sqrt{3} c_2/c_1 \rfloor - w) c_3).$$ 
The projections of $c_1 \mathbb{Z} \oplus c_2 \mathbb{Z} \oplus c_3 \mathbb{Z}$ onto $\hat{\pmb{u}}_w^\perp$ will be, up to equivalence, arbitrarily close to $A_2$ when $w \to \infty$.
\end{ex}

\subsection{Comparisons: Curves in $\mathbb{R}^6$}
\label{exemplo}
Here we compare our approach of constructing curves on torus layers with the curves constructed in \cite{FatStrut}  and \cite{Sueli} in terms of length for given small-ball radii. Given $\delta > 0$, we first consider a set of flat tori associated to a spherical code in $\mathbb{R}^3$, with minimum distance greater than $2 \delta$, as described in Section \ref{sec:torusLayers}. Through the first inequality of \eqref{eq:distortionCurve}, for each torus, we can find $r_{\pmb{c}}$ in order to guarantee that the curves on the flat tori will have small-ball radius at least $\delta$ (this is also done in the case of the curves on the torus $c = \pmb{\hat{e}}$ used in \cite{Sueli}). We then look for the larger element of the sequence of vectors that produces a projection with minimum distance at least $r_{\pmb{c}}$. 

\begin{figure}[htb!]
\centering
\includegraphics[scale=0.50]{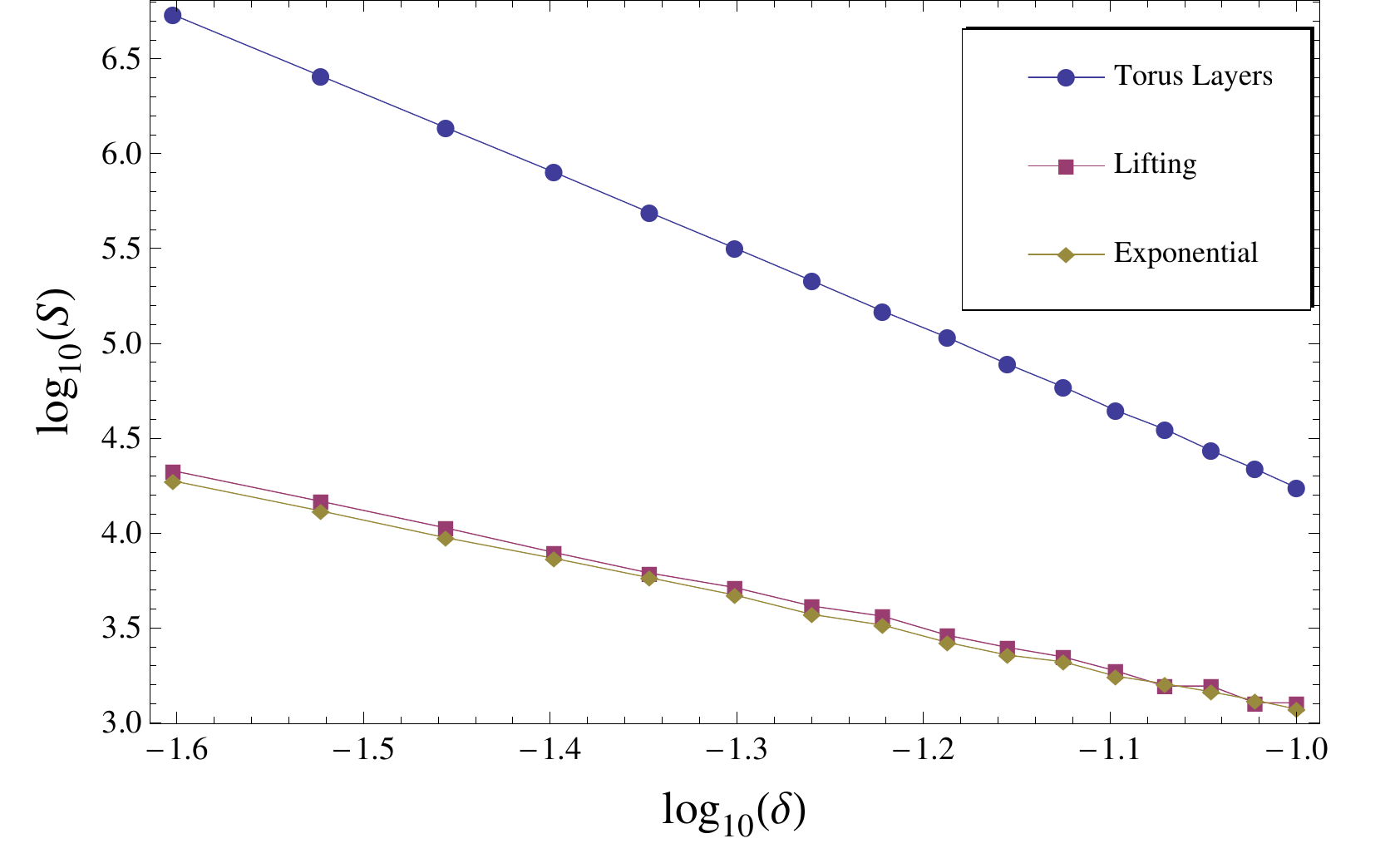}
\caption{Comparison between diferent approaches in terms of small-ball radius $\delta$ and arc-length $\mathcal{S}$. }
\label{grafico}
\end{figure}

In Figure \ref{grafico}, the first curve from the bottom to the top represents the exponential sequence \cite{Sueli}. The second one is obtained by directly applying the Lifting Construction \cite{FatStrut} to the hexagonal lattice and the last one displays the total length associated to our scheme.

\section{Mean Squared Error Analysis}
\label{sec:mse}
The first step of the encoding process described in Section \ref{subsec:Enc} involves a choice of partition of the interval $[0,1]$. If the source to be transmitted is uniformly distributed over $[0,1]$, it is intuitive that the best choice in terms of mse is the one such that the encoding map will have constant stretch. This fact is formalized in the next theorem.

\begin{teo}Let $I_1,\ldots,I_M$ be a partition of the interval $[0,1)$. Let $\{\pmb{s}_1,\ldots,\pmb{s}_M\}$ be a collection of non-overlapping curves on $\mathcal{S}^{2N-1}$ of length $L_1,\ldots,L_M$ and domain $[0,1]$. Suppose that the encoding rule is $\pmb{s}(x) = \pmb{s}_j(f_j(x)), \mbox{ if } x \in I_j,$
where  $f_j:I_j \to [0,1]$ is the linear function that takes $I_j$ to $[0,1]$. Then, under the low noise regime, the minimum $\epsilon^*$ of $E_{\small \mbox{low}}[(X-\hat{X})^2]$ over all partitions $\{I_1,\ldots,I_m\}$ is achieved when

\begin{equation}|I_j| = \displaystyle \dfrac{L_j}{\displaystyle \sum_{j=1}^M L_j} = \frac{L_j}{L}\label{eq:optimalPartition}\end{equation}
and satisfies $\epsilon^* = \sigma^2/L^2$.
\end{teo}
\begin{proof}First note that 
\begin{equation*} \begin{split} E_{\small \mbox{low}}[(X-\hat{X})^2] &= \sum_{j=1}^M E_{\small \mbox{low}}[(X-\hat{X})^2 | X \in I_j] P(X \in I_j) \\
&= \sum_{j=1}^M E_{\small \mbox{low}}[(X-\hat{X})^2 | X \in I_j] |I_j| \end{split}
\end{equation*}
and if the noise is small, $E_{\small \mbox{low}}[(X-\hat{X})^2 | X \in I_j] = \sigma^2 |I_j|^2/L_j^2$. Therefore

\begin{equation*}
\begin{split}
E_{\small \mbox{low}}[(X-\hat{X})^2] &= \sigma^2 \sum_{j=1}^M \frac{|I_j|^3}{L_j^2} = \sigma^2 L \sum_{j=1}^M \left(\frac{|I_j|}{L_j}\right)^3 \left(\frac{L_j}{L}\right) \\& \stackrel{(a)}{\geq} \sigma^2 L   \left( \sum_{j=1}^M\frac{|I_j|}{L}\right)^3 =\frac{\sigma^2}{L^2},
\end{split}
\end{equation*}
where (a) is due to the strict convexity of the function $g(x) = x^3$ for $x > 0$. Equality in (a) only holds when $|I_j|/L_j = |I_k|/L_k, \forall k,j$. This condition, together with $\sum_{j=1}^n |I_j| = 1$, implies Equation \eqref{eq:optimalPartition}.

\end{proof}

It was proved in \cite{Sueli} that for fixed $N$ and $\sigma^2$, the mse of both the ``spherical code'' and the ``shift-map'' schemes decay, with respect to the channel power, with order $O(1/P^N)$, provided that the parameter $a$ \cite[Sec. VI. A]{Sueli} is chosen suitably. This means that, although the surface of the sphere is a $(2N-1)$-dimensional object, the mse exponent with respect to the channel power is roughly half of these dimensions for the spherical code. In what follows we present two results:

\begin{enumerate} 
\item A constructive scheme based on Example \ref{exemplo} that increases of $N!$ the lengths of the scheme proposed in \cite{Sueli}, even with the improvements of the Lifting Construction \cite{FatStrut}. This will yield an asymptotic behavior comparable to $O(1/P^N)$, but with a better performance when $N$ increases.
\item A non-constructive asymptotic argument showing that it is possible to ``recover'' the $(N-1)$-dimensions lost by mapping the $N$-dimensional shift-map to a sphere in $\mathbb{R}^{2N}$ and achieve decay order of $O(1/P^{-(2N-1)})$.
\end{enumerate}

For the first result, let $N$ be fixed and let $\pmb{c}(t)$ be given by Equation \eqref{eq:torosPerm}. Recall that $\pmb{\hat{e}} = 1/\sqrt{N} (1,\ldots,1)$. For $\rho > 0$ sufficiently small, there is $t$ such that the tori in the set $\mathcal{S} \mathcal{C}_{c(t)}$ have minimum distance $2 \rho$. Moreover, there is $\hat{\pmb{u}} \in \Lambda_{\pmb{c}(t)}$ arbitrarily close to the bound \eqref{eq:BoundDensidade}. The same vector $\pmb{u}$ can be used for any $\pmb{c} \in \Sc_{\pmb{c}(t)}$ just by interchanging its coordinates, yielding the same parameters for all tori. Since we have $N!$ tori in the set $\Sc_{\pmb{c}(t)}$, the total length $L_{TL}$ produced is given by:


\begin{equation} L_{TL}  = (N!)L_{\pmb{c}(t)} = N! \,\, \frac{(2\pi)^{N}}{\rho^{N-1}} {\prod_{i=1}^N c_i(t)}\left( \frac{\Delta_{N-1}}{\Vo_{N-1}} - \varepsilon_1 \right)
\end{equation}

For the Lifting Construction and same small-ball radius, we have

$$L_{LC} = L_{\pmb{\hat{e}}} = 2 \pi {\left\|\hat{\pmb{u}} \right\|} = \frac{(2\pi)^{N}}{\rho^{N-1} N^{N/2}}  \left( \frac{\Delta_{N-1}}{\Vo_{N-1}} - \varepsilon_2 \right).$$

As $\rho \to 0$, we can make both $\varepsilon_1$ and $\varepsilon_2$ vanish, and then

$$\frac{L_{TL}}{L_{LC}} \to N! \,\, N^{N/2} {\prod_{i=1}^N c_i(t)}.$$

Now, by making $t$ sufficiently small (but keeping the small-ball radius greater than $\rho$), we have

$$\frac{L_{TL}}{L_{LC}} \to N! \mbox{ as } t \to 0.$$

In other words, 

$$E_{\small \mbox{low}}^{TL}[(X-\hat{X})^2] \approx \frac{1}{(N!)^2} E_{\small \mbox{low}}^{LC}[(X-\hat{X})^2],$$
where the approximation is good when we allow the small-ball radius to be small (this is the case when the SNR is high).

The next theorem states the asymptotic improvements in terms of ``recovering'' the $N-1$ dimensions lost by the spherical code scheme in \cite{Sueli}. Its proof, however, is non-constructive. The intuitive idea is that the spherical uses only one torus (an $N$-dimensional object) on the surface of a sphere, while the whole sphere can be filled through the foliation by flat tori.

\begin{teo}
There exists a choice of parameters such that the mean squared error of the Torus Layers Scheme decays with order $O(P^{-{(2N-1)+\mu}})$ for any arbitrarily small $\mu > 0$.
\label{eq:exponentOrder}
\end{teo}
\begin{proof}
See Appendix A.
\end{proof}

\section{Simulations}
\label{sec:sim}
We now present some simulation results in order to illustrate the performance of our proposed scheme. We consider sources uniformly distributed over the interval $[0,1]$ and normally distributed with variance $\sigma_s^2=1/2$.  As a benchmark we chose the spherical code scheme (V$\&$C) presented in \cite{Sueli}, since this scheme has comparable performance with \cite{polynomial}, as shown in their simulations. We also plotted the curves corresponding to a simple linear modulation.  The simulation results correspond to $50000$ samples with $1$ to $6$ expansion rate ($N=3$) and the performance is measured in terms of $1/\mbox{mse}$ as a function of the $\mbox{SNR} = P/\sigma^2$.

\subsection{Uniform sources}

In Figure \ref{fig:graficoBall} (top), the Torus Layer Scheme (TL) performance for $M = 6$ tori, designed as in Example 1 with $t = 0.6$, and projection vector $\bm{u}= (1,2,198)$ is displayed. The spherical code scheme V$\&$C was simulated for $a = 18$. The parameters were chosen such that both schemes reach the asymptote for the same SNR (i.e., have approximately the same small-ball radius). As expected, our scheme outperforms (V$\&$C) in the low noise regime (about $13$dB higher in this example).

\begin{figure}[htb]
\centering
\includegraphics[scale=0.41, angle=-90]{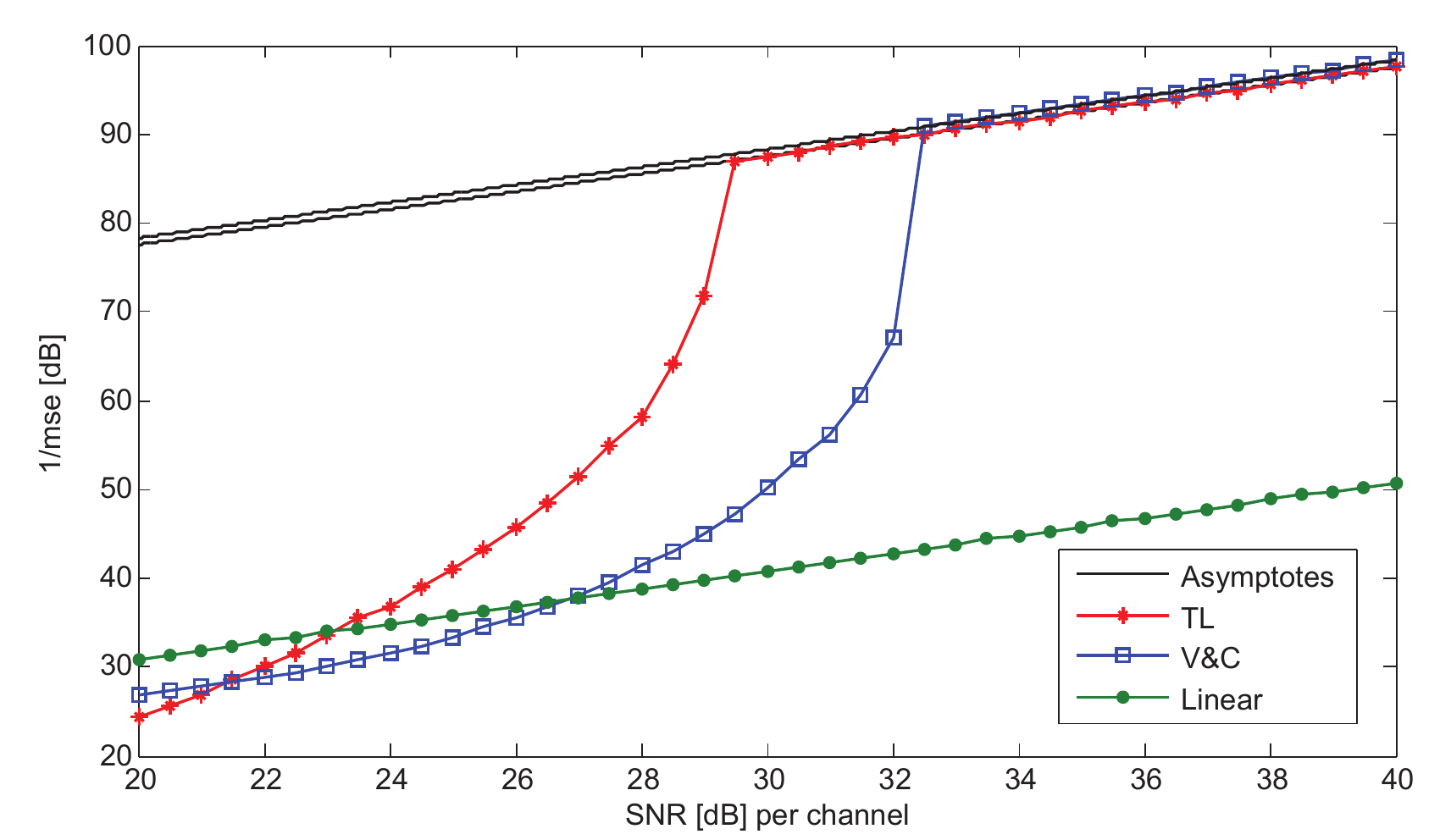}
\vspace{1cm} .
\includegraphics[scale=0.4, angle=-90]{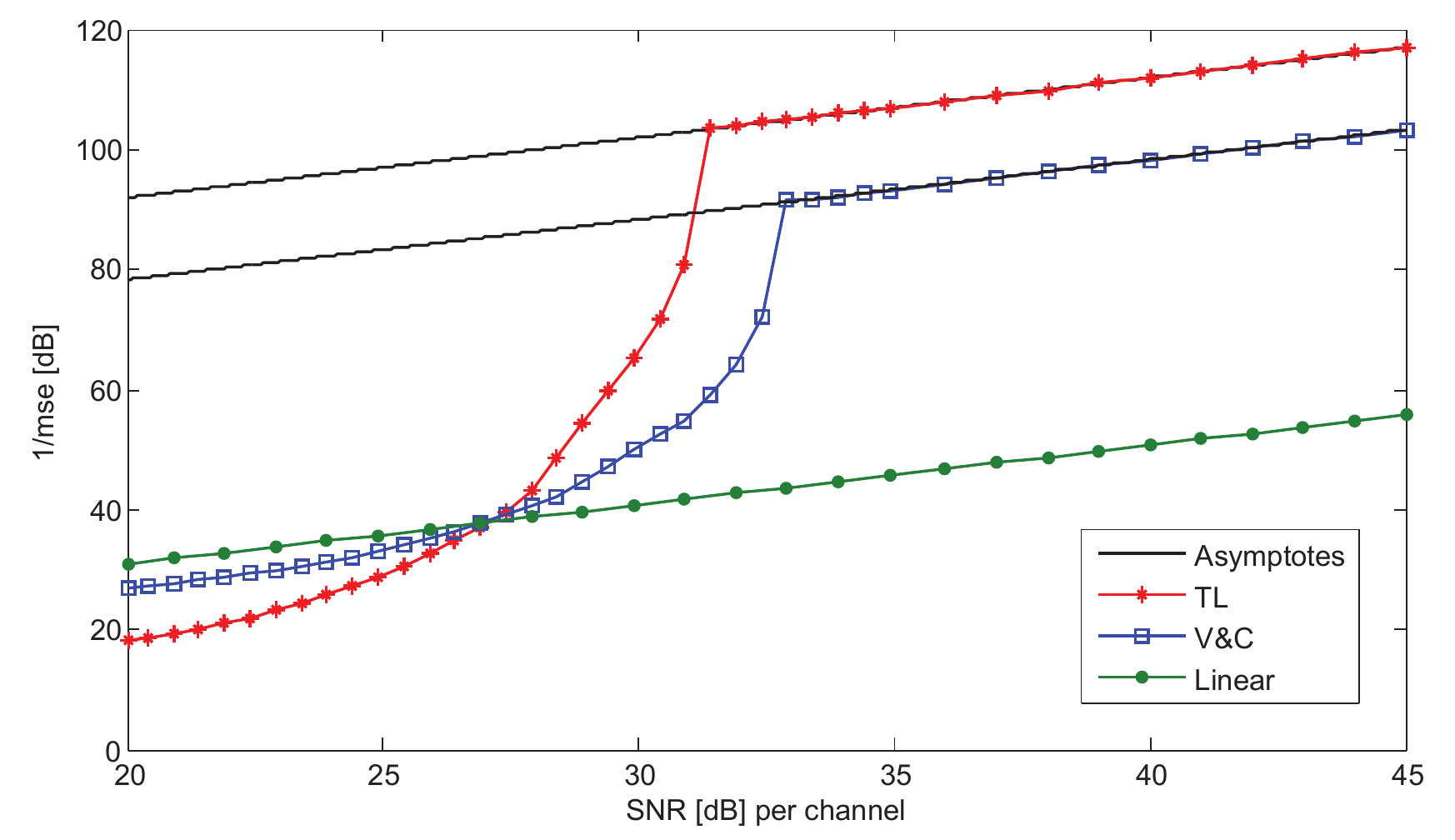}
\caption{Top: both schemes have about the same performance in the low noise regime. Bottom: both schemes reach their asymptote for the same SNR.}
\label{fig:graficoBall}
\end{figure}

Figure \ref{fig:graficoBall} shows, on the bottom, simulation results for $M = 6$ tori, also designed as in Example 1 with $t = 1.75$, and projection vector $\bm{u}= (1,4,34)$. The spherical code scheme V$\&$C was simulated also for $a = 18$. These parameters were chosen such that both schemes achieve about the same performance in low noise regime (same asymptote). In this case, our scheme reaches the asymptote for a smaller SNR ($5$dB in this example). In both graphics it is possible to see that linear modulation is a good choice only for very small SNR.

%

\subsection{Gaussian sources}

Although the theory is suitably developed for uniform sources, we present in this section a possible approach for the Gaussian case and also some simulation results. Let $x$ be drawn according to a normal distribution $\mathcal{N}(0,\sigma_s)$ and let $p(x)$ be the pdf of $x$. In order to use our coding scheme we need to map the real line into sub intervals of $[0,1]$ and then in layers of flat tori. 

Like in the uniform case, we consider a collection $T = \left\{ T_1, \ldots, T_M \right\}$ of $M$ tori designed as described in Section \ref{sec:torusLayers}. For each one of these tori, we have a curve $\CurvaT{k}(x)$ with length $L_k$ and the total length is $L = \sum_{k=1}^M L_k$.

Our first step to deal with Gaussian sources is to split the real line into $M$ sub intervals $\{Q_1, Q_2, \cdots, Q_M\}$ such that the area under $p(x)$ with domain restricted to $Q_k$ is equal to $\dfrac{L_k}{L}$. Let $\{x_1, x_2, \cdots, x_{M-1}\}$ be the endpoints, such that the intervals are defined as follows:
$$
	Q_1 = (-\infty, x_1], \, \, Q_k = (x_{k-1}, x_k],  \, \, Q_M = (x_{M-1},\infty).
$$

\begin{figure}[h!tb]
\centering
\includegraphics[scale=0.5]{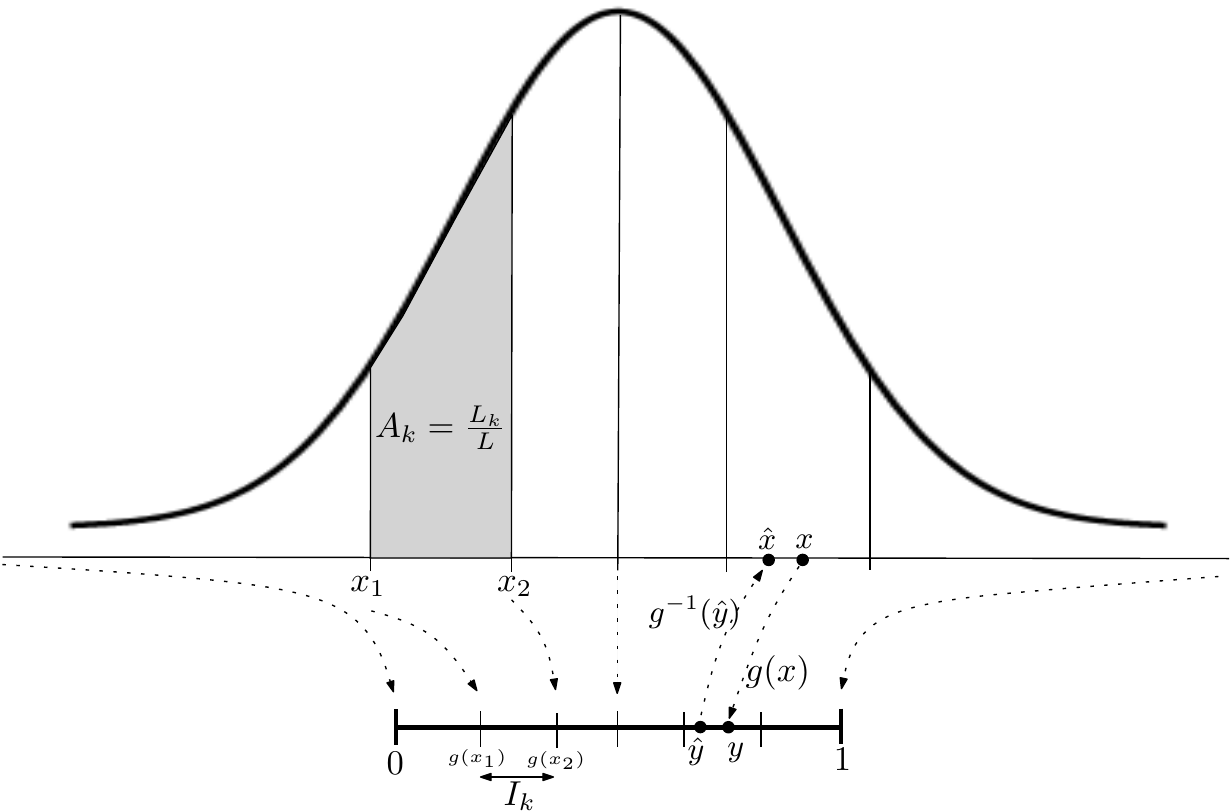}
\includegraphics[scale=0.3]{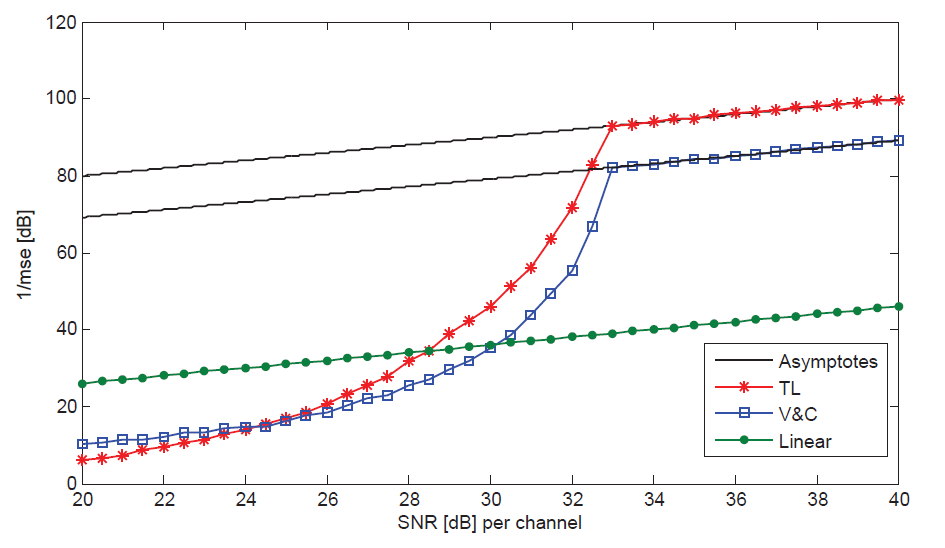}
\caption{Top: Interval splitting for Gaussian sources using the companding function $g(x)$. Bottom: Simulation results for a Gaussian source $\mathcal{N}(0,0.5)$.}
\label{fig:normaldivision}
\end{figure}
\noindent Then we apply the optimal companding function \cite{Sakrison} 
$$g:\R\rightarrow [0,1]$$
\begin{equation}
\displaystyle
	g(x)=\frac {\int_{-\infty}^x p(u)^{\frac{1}{3}} \, du} {\int_{-\infty}^{\infty} p(u)^{\frac{1}{3}} \, du}
\end{equation}
to map the real line into the interval $[0,1]$. The images $\{g(x_1), g(x_2), \cdots, g(x_{M-1}) \}$ induce a split of the interval $[0,1]$ into $M$ pieces $\{I_1, I_2, \cdots, I_M\}$, such that $$
	I_1  = [0, g(x_1)], \, \, I_k = (g(x_{k-1}), g(x_k)], \, \, I_M  = (g(x_{M-1}),1], 
	$$
 with lengths $\{L_{I_1}, L_{I_2}, \cdots, L_{I_M}\}$ and then we can define our bijective mapping for the Gaussian case as
$$\begin{array}{cc}
f_k: I_k  \to [0,1) \\
f_k(y) = \displaystyle\frac{y - \sum_{j=1}^{k-1} L_{I_j}}{L_{I_k}}
\end{array}, \mbox{ where } y=g(x).
$$
The full encoding map $\pmb{s}$ as well as the decoding process will be similar to the uniform case. Once we have the estimate $\hat{y}$ we must invert the companding function to get $\hat{x} = g^{-1}(\hat{y})$ and then compute the mse. This process is illustrated in Figure \ref{fig:normaldivision} (left).

In Figure \ref{fig:normaldivision} (top) we show simulation results for a Gaussian source with $\sigma_s = 0.5$. The parameters were the same as for the Uniform source simulation in Figure \ref{fig:graficoBall} (top). As we can see, the threshold is reached for a higher SNR than in the uniform case, but the TL scheme maintain some advantage over V$\&$C.

We finish this section with a remark on the implementation. Since the curves constructed from our method are closed, if the sent value is one extreme value of the interval $I_j = [L_{j-1}/L, L_j/L]$ it may happen that the decoder wrongly interpret it as the other extreme. In other words, if $L_{j-1}/L + \varepsilon$ is sent, then with a non-negligible probability, the decoder will output $L_j/L$. To prevent these errors, the encoder needs to ``shrink'' the interval $I_j$ by a factor $\alpha \in (0, 1)$, such that the curve is ``opened'', separating the encoded extremes values apart. However, this process deteriorates the low noise performance, so that $\alpha$ needs to be calibrated. In our examples, it sufficed to take any $\alpha$ between $0.7$ and $0.8$. This issue was briefly commented in \cite{Sueli}.

\section{Conclusion}
The problem of transmitting a continuous alphabet source over an AWGN channel was considered through an approach based on curves designed in layers of flat tori on the surface of a $(2N)$-dimensional Euclidean sphere. This approach explores connections with constructions of spherical codes and is related to the problem of finding dense projections of the lattice $c_1 \mathbb{Z} \oplus \ldots \oplus c_N \mathbb{Z}$.

This work is a generalization of both the scheme proposed in \cite{Sueli} and the Lifting Construction in \cite{FatStrut}. As a consequence, our scheme compares favorably to previous works in terms of the tradeoff between total length and small-ball radius, which is a proper figure of merit for this communication system.

In spite of the improvements in terms of length versus small-ball radius, the constructiveness, homogeneity and overall complexity of the decoding algorithm are features preserved from \cite{Sueli}. Simulations also show a good scaling of the mse with the channel SNR. Further questions include applications to sources that do not have limited support, or when it is hard to employ a companding function (e.g., Gaussian mixtures), and generalization to other bandwidth expansion ratios than $1:2N$ (the general case is to consider $K$ to $N$ expansion mappings, when $1 < K < N$). In the more general $K:N$ bandwidth expansion context we must consider $K$-dimensional manifolds in $\mathbb{R}^N$, and the analog codes could be built for instance considering projections of lattices onto $(N-K)$-dimensional subspaces and a generalized version of the shift-map system described in \cite{Sueli}.


\section*{Acknowledgment}

The authors would like to thank the Centre Interfacultaire Bernoulli (CIB) - EPFL, where part of this work was developed, during the special semester on Combinatorial, Algebraic and Algorithmic Aspects of Coding Theory, and Vinay Vaishampayan for helpful discussions on analog source-channel coding. The authors also thank the reviewers for the suggestions that improved the manuscript, specially for an observation that led us to Theorem \ref{eq:exponentOrder}.
\section*{Appendix A}
The big-O and big-$\Theta$ notations are employed throughout this appendix. We say that $f(n) = O(g(n))$ if there exists $n_0$ and a constant $c$ such that $|f(n)| \leq c|g(n)|$ for $n > n_0$. We say that $f(n) = \Theta(g(n))$ if $f(n) = O(g(n))$ and $g(n) = O(f(n))$. We have the following technical lemmas:
\begin{lema} Consider a family of analog codes given by the map \eqref{eq:encoding} normalized to average power $P$. Suppose that the total length and small-ball radius of the family satisfies $L = \Theta(\rho^{-(k-1)})$ for some $k > 1$, as $\rho \to 0$. Then, for any $P > 0$ and any arbitrarily small $\mu > 0$, there is a choice of parameters such that the \mbox{mse} decays with order $O(P^{-k+\mu})$.
\end{lema}
\begin{proof}The proof is an application of the technique in \cite[Appendix A]{Robust} or even  \cite[Sec. VI. B]{Sueli}, that we briefly describe for sake of completeness. If $P(\mbox{``jump''})$ denotes the probability that the estimate $\hat{x}$ is decoded to the wrong fold of the curve, we can bound it by $Q(\sqrt{P} \rho/\sigma)$, where $Q(x)$ is the tail probability of a standard normal distribution. On the other hand, when there is no jump, the mse is proportional to $\sigma^2/(\alpha^2 L^2)$, where $L$ is the length of the curve and ${\alpha = \sqrt{P}}$. Since $\rho \to 0$, there is a choice of parameters such that $\rho = \Theta(P^{-1/2+\bar{\mu}})$, for $\bar{\mu} > 0$ arbitrarily small. If we choose $\bar{\mu} = \mu/(2k-2)$, then $\sqrt{P} \rho = \Theta(P^{\bar{\mu}})$ and from the previous observations:

\begin{equation}
\begin{split}
\mbox{mse} &\leq E[(x-\hat{x})^2 | \mbox{``no jump''}] + P(\mbox{''jump''}) \\ &\leq \underbrace{\displaystyle O(P^{-k+\mu})}_{\displaystyle \frac{\sigma^2}{P L^2}} + \underbrace{\displaystyle Q\left(\frac{\sqrt{P} \rho}{\sigma}\right)}_{\displaystyle \mbox{exponentially small}}.
\end{split}
\end{equation}
\end{proof}

\begin{lema} If $\rho > 0$ and $\mu > 0$ are sufficiently small, there exists a spherical code $\mbox{SC}_{+} = \left\{ \pmb{c}_1, \ldots, \pmb{c}_M \right\}$ with minimum distance $2 \rho$  satisfying:

\begin{equation} \sum_{i=1}^ M{\prod_{j=1}^N c_{ij}} = \Theta(\rho^{-(N-1)+\mu}).
\end{equation}
\end{lema}
\begin{proof} From the Chabauty-Shannon-Wyner bound \cite[Thm. 1.6.2]{eric}, there exists a spherical code with minimum distance $2 \rho$ and number of codewords at least $M = \Theta(\rho^{-(N-1)})$. In fact, by considering only the volume of the positive part of the sphere ($x_i \geq 0$ for all $i$), the same argument shows that there exists a spherical $SC_{+}$ with codewords having positive coordinates and satisfying $| SC_{+} | = \Theta(\rho^{-(N-1)}) $. Now, denote by $SC_{+}^\varepsilon$ the subcode of $SC_{+}$ such that $c_{ij} > \varepsilon$ for some positive small $\varepsilon$. We have $  \sum_{i=1}^ M{\prod_{j=1}^N c_{ij}} \geq \sum_{\pmb{c}_i\in SC_{+}^{\varepsilon}}{\prod_{j=1}^N c_{ij}} \geq \varepsilon^N |SC_{+}^\varepsilon|.$ As $\varepsilon \to 0$, $|SC_{+}^\varepsilon| \to |SC_{+}| = \Theta(\rho^{-(N-1)})$ , however if $\varepsilon$ goes to $0$ too fast with respect to the distance, we may be deleting too many  codewords to construct $SC_{+}^\varepsilon$ and hence will not achieve the desired exponent. The choice $\varepsilon = \rho^{\mu/N}$ for arbitrarily small (but positive) $\mu$ does the trick.

\end{proof}

\textit{Proof of Theorem \ref{eq:exponentOrder}:} If the spherical code $SC_{+} = \left\{ \pmb{c}_1, \ldots, \pmb{c}_M \right\}$ has $M$ words and minimum distance $2 \rho$ sufficiently small, from the Scaled Lifting Construction (Section V) it is possible to find vectors $\pmb{u}_i$ such that their norms are close to the bound \eqref{eq:BoundDensidade}, hence for some $\varepsilon$, the total length will be

\begin{equation} L =  \sum_{i=1}^M 2 \pi {\left\|\hat{\pmb{u}_i} \right\|} = \frac{(2\pi)^{N}}{\rho^{N-1}} \sum_{i=1}^ M{\prod_{j=1}^N c_{ij}}\left( \frac{\Delta_{N-1}}{\Vo_{N-1}} - \varepsilon\right),
\end{equation}
where $c_{ij}$ is the $j$-th coordinate of torus $T_i$ and $\varepsilon \to 0$ as $\rho \to 0$. Thus desconsidering the ``sum of the product'' term, the length would grow with order  $\rho^{-(N-1)}$. From Lemma 2, we can choose $\pmb{c}_i$ such that the the sum has order $\Theta(\rho^{-(N-1)})$, hence $L = \Theta(\rho^{-(2N-2)})$, and from Lemma 1 the result holds. \QED

\bibliographystyle{plain}
\bibliography{curves}

%
%
%
%
%
%

\end{document}